\documentclass[10pt,conference]{IEEEtran}\IEEEoverridecommandlockouts
\usepackage{algorithmic}
\usepackage{algorithm}
\usepackage{amsfonts}
\usepackage{amssymb}
\usepackage{amsmath}
\usepackage{amsthm}
\usepackage{graphicx}
\usepackage{color}
\usepackage{multirow}
\usepackage{cite}
\usepackage{verbatim}
\usepackage{bm}
\usepackage{epstopdf}
\usepackage{caption}
\usepackage{stfloats}

\usepackage{subcaption}
\usepackage[utf8]{inputenc}
\usepackage[english]{babel}
\usepackage{etex,etoolbox}
\usepackage{thmtools}
\usepackage{environ}
\usepackage{mathtools}
\newtheorem{theo}{Theorem}

\newtheorem{coro}{Corollary}

\def\dsum{\mathop{\displaystyle \sum }}%
\begin{document}

\date{\thistime,\,\today}
\title{Multiuser Precoding and Channel Estimation for Hybrid Millimeter Wave MIMO Systems}
\author{\thanks{ D. W. K. Ng is supported under Australian Research Council's Discovery Early Career Researcher Award funding scheme (project
number DE170100137). This work was supported in part by the Australian Research Council (ARC)
Linkage Project LP 160100708.}\IEEEauthorblockN{Lou Zhao, Derrick~Wing~Kwan~Ng, and Jinhong Yuan}
\IEEEauthorblockA{School of Electrical Engineering and
Telecommunications, The University of New South Wales, Sydney, Australia}
Email: lou.zhao@unsw.edu.au, w.k.ng@unsw.edu.au, j.yuan@unsw.edu.au}

\maketitle

\begin{abstract}

In this paper, we develop a low-complexity channel estimation for hybrid millimeter wave (mmWave) systems, where the number of radio frequency (RF) chains  is much less than the number of antennas equipped at each transceiver.
The proposed channel estimation algorithm aims to estimate the strongest angle-of-arrivals (AoAs) at both the base station (BS) and the users.
Then all the users transmit orthogonal pilot symbols to the BS via these estimated strongest AoAs to facilitate the channel estimation.
The algorithm does not require any explicit channel state information (CSI) feedback from the users and the associated signalling overhead of the algorithm is only proportional to the number of users, which is significantly less compared to various existing schemes.
Besides, the proposed algorithm is applicable to both non-sparse and sparse mmWave channel environments.
Based on the estimated CSI, zero-forcing (ZF) precoding is adopted for multiuser downlink transmission. In addition, we derive a tight achievable rate upper bound of the system.
Our analytical and simulation results show that the proposed scheme offer a considerable achievable rate gain compared to fully digital systems, where the number of RF chains equipped at each transceiver is equal to the number of antennas. Furthermore, the achievable rate performance gap between the considered hybrid mmWave systems and the fully digital system is characterized, which provides useful system design insights.

\end{abstract}

\renewcommand{\baselinestretch}{0.93}
\large\normalsize

\section{Introduction}

Higher data rates, large bandwidth, and higher spectral efficiency are necessary for the fifth-generation (5G) wireless communication systems to support various emerging applications \cite{Kwan_5G}. The combination of millimeter wave (mmWave) communication \cite{AZhang2015,Dai2016,Kokshoorn2016,JR:Kwan_massive_MIMO} with massive multiple-input multiple-output (MIMO) \cite{Yang2015,Bogale2015,Marzetta2010} is considered as one of the promising candidate technologies for 5G communication systems with many potential and exciting opportunities for research \cite{Swindlehurst2014,Bogale2015,Deng2015,Sohrabi2016,Rappaport2015,Bjornson2016}.
For example, the trade-offs between system performance, hardware complexity, and energy consumption \cite{AZhang2015,Heath2016a} are still unclear.
From the literature, it is certain that the conventional fully digital MIMO systems, in which each antenna connects with a dedicated radio frequency (RF) chain, are impractical for mmWave systems due to the prohibitively high cost, e.g. tremendous energy consumption of high resolution analog-to-digital convertors/digital-to-analog convertors (ADC/DACs) and power amplifiers (PAs).
Therefore, several mmWave hybrid systems were proposed as compromised solutions which strike a balance between hardware complexity and system performance \cite{Ni2016,Alkhateeb2015,Ayach2014,Han2015,Sohrabi2016}.
Specifically, the use of a large number of antennas, connected with only a small number of independent RF chains at transceivers, is adopted to exploit the large array gain to compensate the inherent high path loss in mmWave channels \cite{Rappaport2015,Hur2016}.  Yet, the hybrid system imposes a restriction on the number of RF chains which  introduces  a  paradigm  shift  in  the design of both resource allocation algorithms and transceiver
signal processing.

Conventionally, pilot-aided channel estimation algorithms are widely adopted for fully digital multiuser (MU) time-division duplex (TDD) massive MIMO systems \cite{Marzetta2010}  operating in sub-$6$ GHz frequency bands.
However, these algorithms cannot be  directly applied to hybrid mmWave systems as the number of RF chains is much small than the number of antennas.
In fact, for the channel estimation in hybrid mmWave systems, the strategies of allocating analog/digital beams to different users and estimating the equivalent baseband channels are still an open area of research \cite{Alkhateeb2015}.
Recently, several improved mmWave channel estimation algorithms were proposed \cite{Kokshoorn2016,Alkhateeb2015}.
The overlapped beam patterns and rate adaptation channel estimation were investigated in \cite{Kokshoorn2016} to reduce the required training time for channel estimation.
Then, the improved limited feedback hybrid channel estimation was proposed \cite{Alkhateeb2015} to maximize the received signal power at each single user so as to reduce the required training and feedback overheads. However, explicit channel state information (CSI) feedback from users is still required for these channel estimation algorithms. In practice, CSI feedbacks may cause system rate performance degradation due to the limited amount of the feedback and the limited  resolution of CSI quantization.
In addition, the CSI feedback also requires exceedingly high consumption of time resource.
Therefore, a low-complexity mmWave channel estimation algorithm, which does not require explicit CSI feedback, is necessary to unlock the potential of hybrid mmWave systems.

In the literature, most of the existing mmWave channel estimation algorithms leverage the sparsity of mmWave channels due to the extremely short wavelength of mmWave \cite{Kokshoorn2016,Alkhateeb2015}.
Generally, in suburban areas or outdoor long distance propagation environment \cite{Hur2016}, the sparsity of mmWave channels can be well exploited.
In practical urban areas (especially in the city center), the number of unexpected scattering clusters increases significantly and mmW communication channels may not be necessarily sparse.
For instance, in the field measurements in Daejeon city, Korea, and the associated ray-tracing simulation \cite{Hur2016}, the angle of arrivals (AoAs) at the BS and the users were observed under the impact of non-negligible scattering clusters.
In addition, existing mmWave channel estimation algorithms \cite{Kokshoorn2016,Alkhat2014,Alkhateeb2015}, which are designed based on the assumption of channel sparsity, may not be applicable to non-sparse mmWave channels.
Indeed, the scattering clusters of mmWave channels due to macro-objects or backscattering from the objects, have a significant impact on system performance and cannot be neglected in the system design.
Therefore, there is an emerging need for a channel estimation algorithm which is applicable to both non-sparse and sparse mmWave channels.

Motivated by aforementioned discussions, we consider a MU hybrid mmWave system.
In particular, we propose and detail a novel non-feedback non-iterative channel estimation algorithm which is applicable to both non-sparse and sparse mmWave channels.
Also, we analyze the achievable rate performance of the mmWave system using ZF precoding based on the estimated equivalent channel information.

\begin{figure}[t]
\centering \includegraphics[width=3.4in]{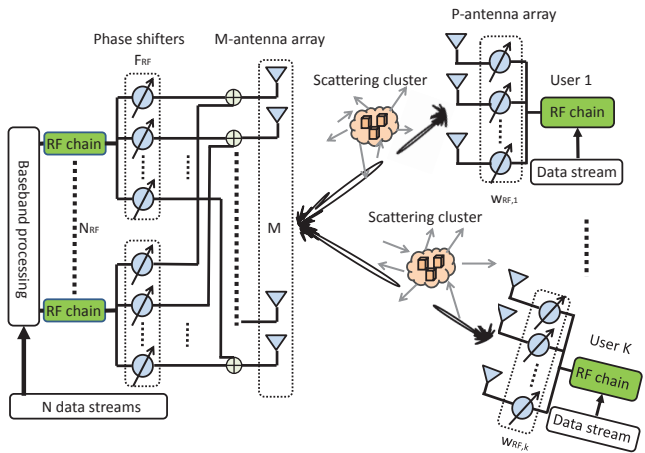}
\vspace{-1mm}
\caption{A mmWave communication system with a hybrid system of
transceivers.}
\label{fig:hybrid}\vspace{-5mm}
\end{figure}

Our main contributions are summarized as follows:

\begin{itemize}
%

\item We propose a three-step MU channel estimation scheme for mmWave channels.
In the first two steps, we estimate the strongest AoAs at both the BS and the users instead of estimating the combination of multiple AoAs.
The estimated strongest AoAs will be exploited for the design of BS and users beamforming matrices.
In the third step, all the users transmit orthogonal pilot symbols to the BS via the beamforming matrices.
We note that the proposed channel estimation scheme does not require explicit CSI feedbacks from the users as well as iterative measurements.
The required training overheads of our proposed algorithm only scale with the number of users. Besides, the proposed algorithm is very general, which is not only applicable to sparse mmWave channels, but also suitable for non-sparse channels.

\item We analyze the achievable rate performance of the proposed ZF precoding scheme based on the estimated CSI of an equivalent channel.
While assuming the equivalent CSI is perfectly known at the BS, we derive a tight performance upper bound of our proposed scheme.
Also, we quantify the performance gap between the proposed hybrid scheme and the fully digital system in terms of achievable rate per user.
It is interesting to note this the performance gap is determined by the strongest AoA component to the scattering component ratio.

\end{itemize}

Notation: $\mathrm{tr}(\cdot )$ denotes trace operation; $\|\cdot\|_{\mathrm{F}}$ denotes the Frobenius norm of matrix; $\lambda _{i}(\cdot )$ denotes the $i$-th maximum eigenvalue of a matrix; $\mathrm{diag}\left\{a\right\} $ is a diagonal matrix with the entries $a$ on its diagonal; $[\cdot ]^{\ast }$ and $[\cdot]^{T}$ denote the complex conjugate and transpose operations, respectively.
The distribution of a circularly symmetric complex Gaussian (CSCG) random vector with a mean vector $\mathbf{x}$ and a covariance matrix ${\sigma}^{2}\mathbf{I}$  is denoted by ${\cal CN}(\mathbf{x},{\sigma}^{2}\mathbf{I})$, and $\sim$ means ``distributed
as". $\mathbb{C}^{N\times M}$ denotes the space of $N\times M$ matrices with complex  entries.

\vspace*{-0mm}

\section{System Model}

We consider a MU hybrid mmWave system which consists of one base station (BS) and $N$ users in a single cell, as shown in Figure \ref{fig:hybrid}.
The BS is equipped with $M\geq 1$ antennas and $N_{\mathrm{RF}}$  radio frequency (RF) chains to serve the $N$ users.
We assume that each user is equipped with $P$ antennas connected to a single RF chain such that $M\geqslant N_{\mathrm{RF}}\geqslant N$.
In the following sections, we set $N = N_{\mathrm{RF}}$ to simplify the analysis.
Each RF chain at the BS can access to all the antennas by using $M$ phase shifters, as shown in Figure \ref{fig:RF_chain}.
\begin{figure}[t]
\centering \includegraphics[width=3.4in]{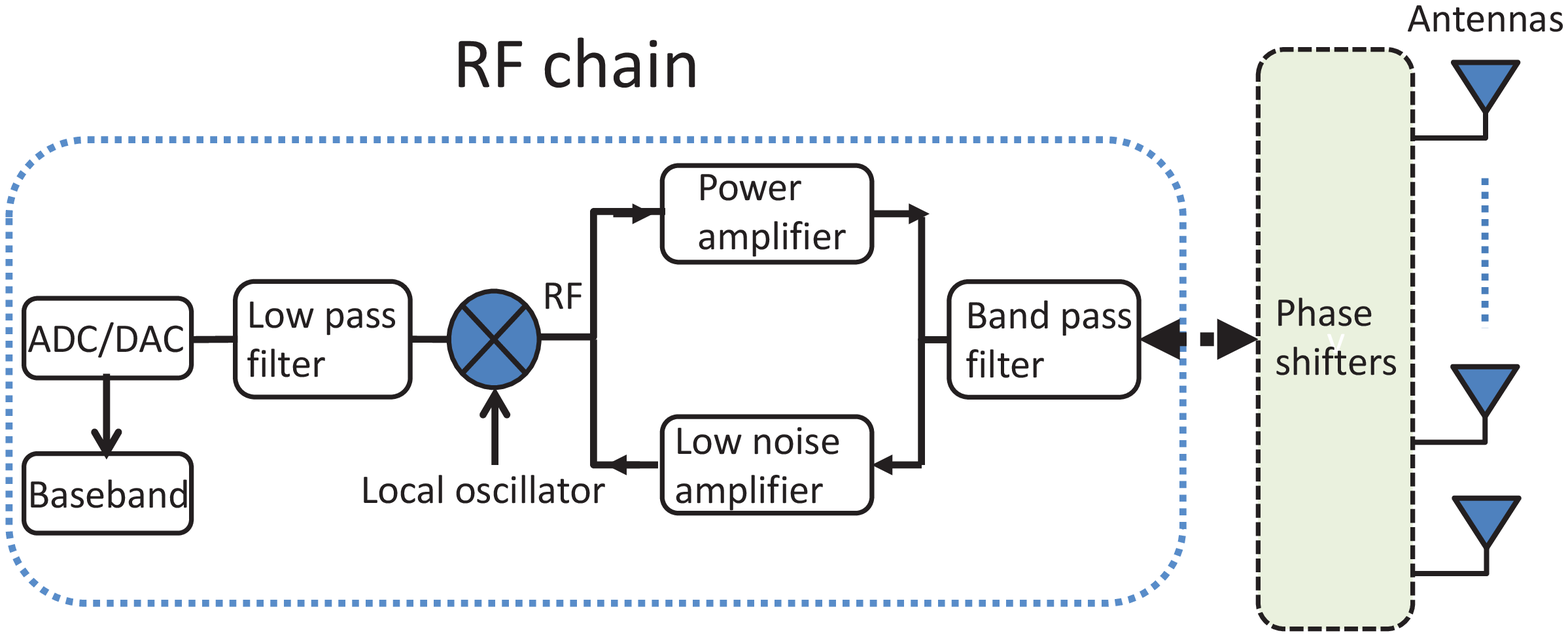}
\vspace{-1mm}
\caption{A block diagram of a RF chain structure.}
\label{fig:RF_chain}\vspace{-5mm}
\end{figure}
At each BS, the number of phase shifters is $M\times N_{\mathrm{RF}}$.
Due to significant propagation attenuation at mmWave frequency, the system is dedicated to cover a small area, e.g. cell radius is $\sim 150$ m.
We assume that the users and the BS are fully synchronized and time division duplex (TDD) is adopted to facilitate uplink and downlink communications \cite{Marzetta2010}.
In previous work \cite{Alkhateeb2015}, mmWave channels were assumed to have sparse propagation paths between the BS and the users.
However, in recent field tests, both a strong line-of-sight (LOS) component and non-negligible scattering component may exist in mmWave propagation channels \cite{Buzzi2016,Rappaport2015,Hur2016}, especially in the urban area.
In particular, mmWave channels can also be modeled by non-sparse Rician fading and with a large Rician K-factor\footnote{We note that existing models \cite{Alkhateeb2015} for sparse mmWave channels are special cases of the considered model.} (approximately $\ge5$ dB) \cite{Hur2016,Rappaport2015,Al-Daher2012}.

Let $\mathbf{H}_{k}\in\mathbb{C}^{M\times P}$ be the uplink channel matrix between the $k$-th user and the BS in the cell.
We assume that $\mathbf{H}_{k}$ is a slow time-varying block Rician fading channel, i.e., the channel is constant in a block but varies slowly from one block to another.
Then, in this paper, we assume that the channel matrix $\mathbf{H}_k$ can be decomposed into a deterministic LOS channel matrix $\mathbf{H}_{\mathrm{L},k}\in\mathbb{C}^{M\times P}$ and a scattered channel matrix $\mathbf{H}_{\mathrm{S,}k}\in\mathbb{C}^{M\times P}$ \cite{Buzzi2016,Dai2016}, i.e., \vspace*{-0mm}
\begin{equation}
\mathbf{H}_{k}=\underset{\mathrm{LOS}\text{\ }\mathrm{component}}{\underbrace{\mathbf{H}_{\mathrm{L,}k}\mathbf{G}_{\mathrm{L,}k}}}+\underset{\mathrm{Scattering}\text{\ }\mathrm{component}}{\underbrace{\mathbf{H}_{\mathrm{S,}k}\mathbf{G}_{\mathrm{S,}k}}},
\end{equation}\vspace*{-0mm}%
where $\mathbf{G}_{\mathrm{L,}k}\in\mathbb{C}^{P\times P}$ and $\mathbf{G}_{\mathrm{S},k}\in\mathbb{C}^{P\times P}$ are diagonal matrices with entries $\mathbf{G}_{\mathrm{L,}k}=\mathrm{diag}\left\{ \sqrt{\frac{\upsilon _{k}}{\upsilon _{k}+1}}\right\}$ and $\mathbf{G}_{\mathrm{S,}k}=\mathrm{diag}\left\{ \sqrt{\frac{1}{\upsilon _{k}+1}}\right\}$, respectively, and $\upsilon _{k}>0$ is the Rician K-factor of user $k$.
In general, we can adopt different array structures, e.g. uniform linear array (ULA) and uniform panel array (UPA) for both the BS and the users.
Here, we adopt the ULA for it is commonly implemented in practice \cite{Alkhateeb2015}.
We assume that all the users are separated by hundreds of wavelengths or more \cite{Marzetta2010}.
Thus, we can express the deterministic LOS channel matrix $\mathbf{H}_{\mathrm{L},k}$ of the $k$-th user as \cite{book:wireless_comm}\vspace*{-0mm}
\begin{equation}
\mathbf{H}_{\mathrm{L,}k}=\mathbf{h}_{\mathrm{L,}k}^{\mathrm{BS}}\mathbf{h}_{%
\mathrm{L,}k}^{H},
\end{equation}\vspace*{-0mm}%
where $\mathbf{h}_{\mathrm{L},k}^{\mathrm{BS}}$ $\in\mathbb{C}^{M\times 1}$ and $\mathbf{h}_{\mathrm{L,}k}$ $\in\mathbb{C}^{P\times 1}$ are the antenna array response vectors of the BS and the $k$-th user respectively.
In particular, $\mathbf{h}_{\mathrm{L,}k}^{\mathrm{BS}}$ and $\mathbf{h}_{\mathrm{L,}k}$ can be expressed as \cite{book:wireless_comm,Trees2002}\vspace*{-0mm}
\begin{align}
\mathbf{h}_{\mathrm{L,}k}^{\mathrm{BS}}=\left[
\begin{array}{ccc}
1, & \ldots
, & \text{ }e^{-j2\pi \left( M-1\right) \tfrac{d}{\lambda }\cos \left(
\theta _{k}\right) }%
\end{array}%
\right] ^{T} \\
\text{and} \text{\ }\mathbf{h}_{\mathrm{L,}k}=\left[\begin{array}{ccc}
1, & \ldots ,
& \text{ }e^{-j2\pi \left( M-1\right) \tfrac{d}{\lambda }\cos \left( \phi
_{k}\right) }%
\end{array}%
\right]^{T},
\end{align}\vspace*{-0mm}%
respectively, where $d$ is the distance between the neighboring antennas and $\lambda $ is the wavelength of the carrier frequency.
Variables $\theta _{k}\in \left[ 0,+\pi \right]$ and $\phi _{k}\in \left[ 0,+\pi \right] $ are the angles of incidence of the LOS path at antenna arrays of the BS and user $k$, respectively.
For convenience, we set $d=\frac{\lambda }{2}$ for the rest of the paper which is an assumption commonly adopted in the literature \cite{Trees2002,book:wireless_comm}.

Without loss of generality, we assume that the scattering component $\mathbf{H}_{\mathrm{S,}k}$ consists $N_{\mathrm{cl}}$ clusters and each cluster contributes $N_{\mathrm{l}}$ propagation path \cite{Buzzi2016}, which can be expressed as
\begin{eqnarray}
\mathbf{H}_{\mathrm{S,}k}&=&\sqrt{\tfrac{1}{{\sum }_{i=1}^{N_{\mathrm{cl}}}{N_{\mathrm{l},i}}}}\overset{N_{\mathrm{cl}}}{\underset{i=1}{\dsum }}\overset{N_{\mathrm{l},i}}{\underset{l=1}{\dsum }}{\alpha _{i,l}}\mathbf{h}_{i,l}^{\mathrm{BS}}\mathbf{h}_{k,i,l}^{H}\notag\\
&=& \left[ \begin{array}{ccc} \mathbf{h}_{\mathrm{S},1},\ldots,\mathbf{h}_{\mathrm{S},P} \end{array}\right],
\end{eqnarray}
where $\mathbf{h}_{i,l}^{\mathrm{BS}}\in\mathbb{C}^{M\times 1}$ and $\mathbf{h}_{k,i,l}\in\mathbb{C}^{P\times 1}$ are the antenna array response vectors of the BS and the $k$-th user associated to the $\left( i,l\right) $-th propagation path, respectively.
Here, $\alpha _{i,l}\sim \mathcal{CN}\left( 0,1\right) $ represents the path attenuation of the $\left(i,l\right)$-th propagation path and $\mathbf{h}_{\mathrm{S},k}\in\mathbb{C}^{M\times 1}$ is the $k$-th column vector of $\mathbf{H}_{\mathrm{S},k}$.
With the increasing number of clusters, the path attenuation coefficients and the AoAs between the users and the BS become randomly distributed \cite{Buzzi2016,Hur2016}.
Therefore, we model the entries of scattering component $\mathbf{H}_{\mathrm{S,}k}$ in a general manner as an independent and identically distributed (i.i.d.) random variable\footnote{To facilitate the study of the downlink hybrid precoding, we assume perfect long-term power control is performed to compensate for path loss and shadowing at the desired users and equal power allocation among different data streams of the users\cite{Alkhateeb2015,Ni2016,Yang2013}. Thus, the entries of scattering component $\mathbf{H}_{\mathrm{S,}k}$ are modeled by i.i.d. random variables. } $\mathcal{CN}\left( 0,1\right)$.
\vspace*{+0.0mm}
\section{Proposed Hybrid Channel Estimation}
In practice, the hybrid system imposes a fundamental challenge for mmWave channel estimation.
Unfortunately, the conventional pilot-aided channel estimation for fully digital systems, e.g. \cite{Kokshoorn2016,Alkhateeb2015}, is not applicable to the considered hybrid mmWave system.
The reasons are that the number of RF chains is much smaller than the number of antennas equipped at the BS and the transceiver beamforming matrix cannot be acquired.

To address this important issue, we propose a new pilot-aided hybrid channel estimation, which mainly contains three steps as shown in Figure \ref{fig:CEI}.
In the first and second steps, we introduce unique unmodulated frequency tones for strongest AoAs estimation, which is inspired by signal processing in monopulse radar and sonar systems \cite{book:wireless_comm}.
The estimated strongest AoAs at both the BS and the users sides will be used for the design of BS and users beamforming matrices.
In the third step, orthogonal pilot sequences are transmitted from all the users to the BS to estimate the uplink channels, which will be adpoted for the design of the BS digital baseband downlink precoder by exploiting the reciprocity between the uplink and downlink channels.

\begin{figure}[t]
\centering
\includegraphics[width =
3.2in]{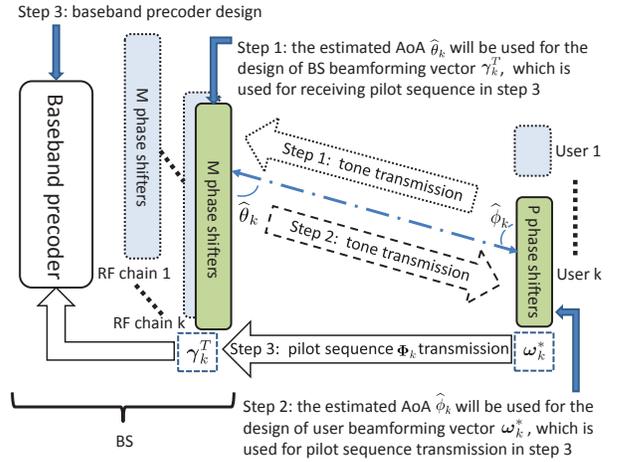}\vspace{-0mm}
\caption{An illustration of the proposed channel estimation algorithm for hybrid mmWave systems.}
\label{fig:CEI}\vspace*{-5mm}
\end{figure}

\subsubsection{Step 1}
First, all the users transmit unique frequency tones to the desired BS in the uplink simultaneously.
For the $k$-th user, an unique unmodulated frequency tone, $x_{k}=\cos \left( 2\pi f_{k}t\right), k\in \{1,\cdots ,N\}$, is transmitted from one of the omni-directional antennas in the antenna array to the BS.
Here, $f_{k}$ is the single carrier frequency and $t$ stands for time and $f_{k}\neq f_{j},  \forall k\neq j$.
For the AoA estimation, if the condition $\frac{f_{k}-f_{j}}{f_{c}} < 10^{-4}, \text{ }\forall k\neq j$ is satisfied, the AoA estimation error by using ULA with a single tone is generally negligible \cite{Trees2002}, where $f_{c}$ is the system carrier frequency.
The pass-band received signal of user $k$ at the BS, $\mathbf{y}_{k}^{\mathrm{BS}}$, is given by\vspace*{-0mm}
\begin{equation}
\mathbf{y}_{k}^{\mathrm{BS}}=\left( \sqrt{\dfrac{\upsilon _{k}}{%
\upsilon _{k}+1}}\mathbf{h}_{\mathrm{L},k}^{\mathrm{BS}}+\sqrt{\dfrac{1}{%
\upsilon _{k}+1}}\mathbf{h}_{\mathrm{S},k}\right) x_{k}+\mathbf{z}_{\mathrm{%
BS}}\mathbf{,}
\end{equation}\vspace*{-0mm}%
where $\mathbf{z}_{\mathrm{BS}}$ denotes the thermal noise at the antenna array of the BS, $\mathbf{z}_{\mathrm{BS}}\sim\mathcal{CN}\left( \mathbf{0},{\sigma_{\mathrm{BS}}^{2}}\mathbf{I} \right)$, and ${\sigma} _{\mathrm{BS}}^{2}$ is the noise variance at each antenna of the BS.

To facilitate the estimation of AoA, we perform a linear angular domain search in $\left[0,180\right]$ with an angle search step size of $\frac{180}{J}$.
Therefore, the AoA detection matrix $\mathbf{\Gamma }_{k}\in\mathbb{C}^{M\times J}$, $\mathbf{\Gamma }_{k}= \left[ \begin{array}{ccc} \bm{\gamma }_{k,1},\ldots,\bm{\gamma }_{k,J} \end{array}\right]$ contains $J$ column vectors. The $i$-th vector $\bm{\gamma }_{k,i}\in\mathbb{C}^{M\times 1}, i\in \{1,\cdots,J\}$, stands for a potential AoA of user $k$ at the BS and is given by\vspace*{-0mm}%
\begin{equation}
\mathbf{\gamma }_{k,i}=
\frac{1}{\sqrt{M}}\left[
\begin{array}{ccc}
1, & \ldots , & \text{ }e^{j2\pi \left( M-1\right) \tfrac{d}{\lambda }\cos
\left( \widehat{\theta }_{i}\right) }
\end{array}%
\right] ^{T},
\end{equation}\vspace*{-0.0mm}%
where $\widehat{\theta }_{i}= \left(i-1\right)\frac{180}{J}, i\in \{1,\cdots,J\}$, is the assumed AoA and $\bm{\gamma }_{k,i}^{H}\bm{\gamma }_{k,i}=1$.
For the AoA estimation of user $k$, $\mathbf{\Gamma }_{k}$ is implemented in the $M$ phase shifters connected by the $k$-th RF chain.
The local oscillator (LO) of the $k$-th RF chain at the BS generates the same carrier frequency $f_{k}$ to down convert the received signals to the baseband, as shown in Figure \ref{fig:RF_chain}.
After the down-conversion, the signals will be filtered by a low pass filter which can remove other frequency tones.
The equivalent received signal at the BS from user $k$ at the $i$-th potential AoA is given by\vspace*{-0mm}
\begin{align}
r_{k,i}^{\mathrm{BS}}=& \sqrt{\dfrac{\upsilon _{k}}{\upsilon
_{k}+1}}\bm{\gamma }_{k,i}^{T}\mathbf{h}_{\mathrm{L},k}^{\mathrm{BS}}\notag \\
&+\sqrt{\dfrac{1}{\upsilon _{k}+1}}\bm{\gamma }_{k,i}^{T}\mathbf{h}_{\mathrm{S},k} +\bm{\gamma }_{k,i}^{T}\mathbf{z}_{\mathrm{BS}
}.  \label{AoA_D}
\end{align}\vspace*{-0mm}%
The potential AoA, which leads to the maximum value among the $J$ observation directions, i.e.,\vspace*{-0mm}%
\begin{equation}
\widetilde{\mathbf{\gamma}}_{k}=\underset{\forall \gamma_{k,i},\text{\ } i\in \{1,\cdots,J\}}{\arg \max \left\vert r_{k,i}^{\mathrm{BS}}\right\vert },  \label{RFBF}
\end{equation}\vspace*{-0mm}%
is considered as the AoA of user $k$.
Besides, vector $\widetilde{\mathbf{\gamma}}_{k}$ corresponding to the AoA with maximal value in (\ref{RFBF}) will be exploited as the $k$-th user's beamforming vector at the BS.
As a result, we can also estimate all other users' uplink AoAs at the BS from their corresponding transmitted signals simultaneously.
For notational simplicity, we denote $\mathbf{F}_{\mathrm{RF}}=\left[\begin{array}{ccc}\widetilde{\mathbf{\gamma}}_{1},\ldots, \widetilde{\mathbf{\gamma}}_{N}\end{array}\right] \in\mathbb{C}^{M\times N}$ as the BS beamforming matrix.
\vspace*{+0mm}%
\subsubsection{Step 2}

The BS sends orthogonal frequency tones to all the users exploiting beamforming matrix $\mathbf{F}_{\mathrm{RF}}$ obtained in step $1$.
This facilitates the downlink AoAs estimation at the users and this information will be used to design the beamforming vectors to be adopted at the users.

The received signal $\mathbf{y}_{k}^{\mathrm{UE}}$ at user $k$ can be expressed as\vspace*{-0mm}
\begin{equation}
\mathbf{y}_{k}^{\mathrm{UE}}=\left[ \mathbf{G}_{\mathrm{L,}k}\mathbf{h}_{\mathrm{L,}k}^{\ast }\left(\mathbf{h}_{\mathrm{L},k}^{\mathrm{BS}}\right)^{T}+\mathbf{G}_{\mathrm{S,}k}\mathbf{H}_{\mathrm{S,}k}^{T}\right]\widetilde{\mathbf{\gamma }}_{k}x_{k} +\mathbf{z}_{\mathrm{MS}},
\end{equation}\vspace*{-0mm}%
where $\mathbf{z}_{\mathrm{MS}}$ denotes the thermal noise at the antenna array of the users, $\mathbf{z}_{\mathrm{MS}}\sim \mathcal{CN}\left( \mathbf{0},{\sigma_{\mathrm{MS}}^{2}}\mathbf{I}\right),$ and ${\sigma} _{\mathrm{MS}}^{2}$ is the noise variance for all the users.

The AoA detection matrix for user $k$, $\mathbf{\Omega }_{k}\in\mathbb{C}^{P\times J}$, which also contains $J$ estimation column vectors, is implemented at phase shifters of user $k$.
The $i$-th column vector of matrix $\mathbf{\Omega }_{k}$ for user $k$, $\mathbf{\omega }_{k,i}\in\mathbb{C}^{P\times 1}, i\in \{1,\cdots,J\}$, is given by\vspace*{-0mm}
\begin{equation}
\mathbf{\omega }_{k,i}=\frac{1}{\sqrt{P}}\left[
\begin{array}{ccc}
1, & \ldots , & e^{j2\pi \left( P-1\right) \tfrac{d}{\lambda }\cos \left(\widehat{\phi }_{i}\right) }%
\end{array}%
\right] ^{T},
\end{equation}\vspace*{-0mm}%
where $\widehat{\phi }_{i}= \left(i-1\right)\frac{180}{J}, i\in \{1,\cdots,J\}$, is the $i$-th potential AoA of user $k$ and $\mathbf{\omega }_{k,i}^{H}\mathbf{\omega }_{k,i}=1$.
With similar  procedures as shown in step 1, the equivalent received signal from the BS at user $k$ of the $i$-th potential AoA is given by \vspace*{-0.0mm}
\begin{align}
r _{k,i}^{\mathrm{UE}}= \text{\ }&\mathbf{\omega }_{k,i}^{H}\sqrt{\dfrac{\upsilon _{k}}{\upsilon
_{k}+1}}\mathbf{h}_{\mathrm{L,}k}^{\ast }\left(\mathbf{h}_{\mathrm{L},k}^{\mathrm{BS}}\right)^{T}\widetilde{\mathbf{\gamma }}_{k}\label{AoA_D2} \\
&+\mathbf{\omega }_{k,i}^{H}\sqrt{\dfrac{1}{\upsilon _{k}+1}}\mathbf{H}_{\mathrm{S,}k}^{T}\widetilde{\mathbf{\gamma }}_{k}+\mathbf{\omega }_{k,i}^{H}\mathbf{z}_{\mathrm{MS}}. \notag
\end{align}\vspace*{-0.0mm}%
Similarly, we search for the maximum value among $J$ observation directions and design the beamforming vector based on the estimated AoA of user $k$.
The beamforming vector for user $k$ is given by\vspace*{-0.0mm}
\begin{equation}
\widetilde{\mathbf{\omega }}_{k}^{\ast}=\underset{\forall \mathbf{\omega }_{k,i},\text{\ }  i\in \{1,\cdots,J\}}{\arg \max \left\vert r
_{k,i}^{\mathrm{UE}}\right\vert }
\end{equation}\vspace*{-0.0mm}and we denote the matrix $\mathbf{Q}_{\mathrm{RF}}=\left[
\begin{array}{ccc}
\widetilde{\mathbf{\omega }}_{1}^{\ast}, \ldots, \widetilde{\mathbf{\omega }}_{N}^{\ast }\end{array}\right] \in\mathbb{C}^{P\times N}$ as the users beamforming matrix.
\vspace*{+0mm}
\subsubsection{Step 3}
The BS and users beamforming matrices based on estimated uplink AoAs and downlink AoAs are designed via step 1 and step 2, respectively.
After that, all the users transmit orthogonal pilot sequences to the BS via user beamforming vectors $\widetilde{\mathbf{\omega }}_{k}^{\ast}$.

We denote the pilot sequences of the $k$-th user in the cell as $\mathbf{\Phi }_{k}=\left[ \vartheta _{k}\left( 1\right)
,\vartheta _{k}\left( 2\right) ,....\vartheta _{k}\left( N\right) \right]^{T}$, $\mathbf{\Phi }_{k}\in\mathbb{C}^{N\times 1}$, stands for $N$ symbols transmitted across time.
The pilot symbols used for the equivalent channel\footnote{The equivalent channel composes of the BS beamforming matrix, the mmWave channel, and the users beamforming matrix.} estimation are transmitted in sequence from symbol $\vartheta _{k}\left( 1\right)$ to symbol $\vartheta _{k}\left(
N\right)$.
The pilot symbols for all the $N$ users form a matrix, $\mathbf{\Psi \in\mathbb{C}}^{N\times N}\mathbf{,}$ where $\mathbf{\Phi }_{k}$ is a column vector of
matrix $\mathbf{\Psi }$ given by\vspace*{-0mm}
\begin{equation}
\mathbf{\Psi }=\ \sqrt{E_{\mathrm{P}}}\left[
\begin{array}{cccc}
\mathbf{\Phi }_{1} & \mathbf{\Phi }_{2} & .... & \mathbf{\Phi }_{N}%
\end{array}%
\right],
\end{equation}\vspace*{-0mm}%
where $\mathbf{\Phi }_{i}^{H}\mathbf{\Phi }_{j}=0,\text{ }\forall
i\neq j,\text{ }i,\text{ }j\in \left\{ 1,\cdots N\right\}$, and $E_{\mathrm{P}}$ represents the transmitted pilot symbol energy.
Note that $\mathbf{\Psi }^{H}\mathbf{\Psi }=E_{\mathrm{P}}\mathbf{I}_{N}$.
Meanwhile, the BS beamforming matrix $\mathbf{F}_{\mathrm{RF}}$ is utilized to receive pilot sequences at all the RF chains.
As the length of the pilot sequences is equal to the number of users, we obtain an $N \times N$ observation matrix from all  the RF chains at the BS.
In particular, the received signal at the $k$-th RF chain at the BS is $\mathbf{s}_{k}^{T}\in\mathbb{C}^{1\times N}$, which is given by\vspace*{-0mm}%
\begin{equation}
\mathbf{s}_{k}^{T}  =\widetilde{\mathbf{\gamma}}_{k}^{T}\overset{N}{\underset{i=1}%
{\sum }}\mathbf{H}_{i}\widetilde{\mathbf{\omega }}_{i}^{\ast }\sqrt{E_{%
\mathrm{P}}}\mathbf{\Phi }_{i}^{T}+\widetilde{\mathbf{\gamma}}_{k}^{T}\mathbf{Z},\label{RSIRFC}
\end{equation}\vspace*{-0mm}%
where $\mathbf{Z}\in\mathbb{C}^{M\times N}$ denotes the additive white Gaussian noise matrix at the BS and the entries of $\mathbf{Z}$ are modeled by i.i.d. random variable with distribution $\mathcal{CN}\left( 0,\sigma _{\mathrm{BS}}^{2}\right)$.

\setcounter{equation}{19}
\begin{figure*}[t]
\begin{theo}\label{thm:Theo_1}
 The achievable rate per user of the ZF precoding is bounded by \vspace*{-0mm}%
\begin{align}
R_{\mathrm{HB}} \leqslant R_{\mathrm{HB}}^{\mathrm{upper}}=\log _{2}\left[ 1+\frac{1}{N^{2}}\left[ \left( \dfrac{\upsilon }{\upsilon+1}\right)MP  \| \mathbf{F}_{\mathrm{RF}}^{H}\mathbf{F}_{\mathrm{RF}} \|_{\mathrm{F}}^{2} +\left( \dfrac{1}{\upsilon +1}\right) N^{2}\right] \dfrac{E_{s}}{\sigma _{\mathrm{MS}}^{2}}\right].  \label{Theo_1}
\end{align}
\end{theo}\vspace*{-1mm}
\begin{proof}
Please refer to Appendix A.
\end{proof}
\hrule\vspace*{-1mm}
\end{figure*}
\begin{figure*}[t]
\begin{coro}
In the large numbers of antennas regime, i.e., $M\rightarrow\infty $, such that $\mathbf{F}_{\mathrm{RF}}^{H}\mathbf{F}_{\mathrm{RF}}\overset{a.s.}{\rightarrow }$ $\mathbf{I}_{N},$ the asymptotic achievable rate per user of the hybrid system is bounded by
\begin{equation}
R_{\mathrm{HB}}^{\mathrm{upper}}\underset{M\rightarrow \infty }{\overset{a.s.}{\rightarrow }}\log _{2}\left\{ 1+\left[
\frac{MP}{N}\left( \dfrac{\upsilon  }{\upsilon +1}\right) +\dfrac{1}{\upsilon  +1}\right] \dfrac{E_{s}}{\sigma _{\mathrm{MS}}^{2}}\right\}.
\label{HSUB_LA}
\end{equation}\label{Coro_1}
\end{coro}
\begin{proof}
The result follows by substituting $\mathbf{F}_{\mathrm{RF}}^{H}\mathbf{F}_{\mathrm{RF}}\underset{M\rightarrow \infty }{\overset{a.s.}{\rightarrow }}\mathbf{I}_{N}$ into (\ref{Theo_1}).
\end{proof}
\hrule\vspace*{-1mm}
\end{figure*}

After $\left[\begin{array}{ccc}\mathbf{s}_{1}, \cdots, \mathbf{s}_{N}\end{array}\right]$ is obtained, we then adopt the least square (LS) method for our equivalent channel estimation.
We note here, the LS method is widely used in practice since it does not require any prior channel information.
Subsequently, with the help of orthogonal pilot sequences, we can construct an equivalent hybrid uplink channel matrix $\mathbf{H}_{\mathrm{eq}}\in\mathbb{C}^{N\times N}$ formed by the proposed scheme via the LS estimation method.
Then, due to the channel reciprocity, the equivalent downlink channel of the hybrid system $\mathbf{H}_{\mathrm{eq}}^{T}$ can be expressed as:\vspace*{-0mm}
\setcounter{equation}{15}
\begin{align}
\widehat{\mathbf{H}}_{\mathrm{eq}}^{T}&=\mathbf{\Psi }^{H}\left[\begin{array}{ccccc}\mathbf{s}_{1} & \ldots & \mathbf{s}_{k} & \ldots & \mathbf{s}_{N}\end{array}\right]\label{EHC_1}  \\
&=\underset{\mathbf{H}_{\mathrm{eq}}^{T}}{\underbrace{\left[
\begin{array}{c}
\widetilde{\mathbf{\omega }}_{1}^{H}\mathbf{H}_{1}^{T}\mathbf{F}_{\mathrm{RF}%
} \\
\vdots \\
\widetilde{\mathbf{\omega }}_{N}^{H}\mathbf{H}_{N}^{T}\mathbf{F}_{\mathrm{RF}}
\end{array}%
\right]}} +\underset{\mathrm{effictive}\text{ }\mathrm{noise}}{\underbrace{%
\frac{1}{\sqrt{E_{\mathrm{P}}}}\left[
\begin{array}{c}
\mathbf{\Phi}_{1}^{H}\mathbf{Z}^{T}\mathbf{F}_{\mathrm{RF}} \\
\vdots \\
\mathbf{\Phi} _{N}^{H}\mathbf{Z}^{T}\mathbf{F}_{\mathrm{RF}}%
\end{array}%
\right] }}.\notag
\end{align}\vspace*{-0mm}%
From Equation (\ref{EHC_1}), we  observe that the proposed hybrid channel estimation can obtain all users' equivalent CSI simultaneously.
Compared to existing channel estimation methods, e.g. compressed-sensing algorithm \cite{Alkhateeb2015}, the proposed algorithm does not require explicit CSI feedback from each antenna equipped at the desired users.
\vspace*{+0.0mm}
\section{Hybrid ZF Precoding and Performance Analysis}
In this section, we illustrate and analyze the achievable rate performance per user of the considered hybrid mmWave system under ZF downlink transmission.
The ZF downlink precoding is based on the estimated hybrid equivalent channel $\mathbf{H}_{\mathrm{eq}}$, which subsumes the BS beamforming matrix $\mathbf{F}_{\mathrm{RF}}$ and the users beamforming matrix $\mathbf{Q}_{\mathrm{RF}}$.
We derive a closed-form upper bound of achievable rate per user of ZF precoding in hybrid mmWave systems.
Also, we compare the system achievable rate upper bound obtained by the fully digital system exploiting ZF precoding for a large number of antennas.

\vspace*{-0.0mm}

\subsection{ZF Precoding}

Now, we utilize the estimated equivalent channel for downlink ZF precoding.
To study the best achievable rate performance of the proposed scheme, we first assume that the equivalent channel is estimated in the high signal-to-noise ratio (SNR) regime, e.g. $E_{\mathrm{P}}\rightarrow \infty$.
Then, the baseband digital ZF precoder $\overline{\mathbf{W}}_{\mathrm{eq}}\in\mathbb{C}^{N\times N}$ based on $\mathbf{H}_{\mathrm{eq}}$ is given by \vspace*{-0.0mm}
\begin{equation}
\overline{\mathbf{W}}_{\mathrm{eq}}=\mathbf{H}_{\mathrm{eq}}^{\ast }(\mathbf{H}_{\mathrm{eq}}^{T}\mathbf{H}_{\mathrm{eq}}^{\ast})^{-1}=\left[\begin{array}{ccc}\overline{\mathbf{w}}_{\mathrm{eq,}1},\ldots ,\overline{\mathbf{w}}_{\mathrm{eq,}N}
\end{array}%
\right] ,  \label{P1}
\end{equation}\vspace*{-0.0mm}where $\overline{\mathbf{w}}_{\mathrm{eq,}k}\in\mathbb{C}^{N\times 1}$ is the $k$-th column of ZF precoder for user $k$.
As each user is equipped with only one RF chain, one superimposed signal is received at each user at each time instant with hybrid transceivers.
The received signal at user $k$ after receive beamforming can be expressed as: \vspace*{-0.0mm}
\begin{align}
y_{\mathrm{ZF}}^{k}&=\underset{\mathrm{desired}\text{ }\mathrm{signal}}{%
\underbrace{\widetilde{\mathbf{\omega }}_{k}^{H}\mathbf{H}_{k}^{T}\mathbf{F}%
_{\mathrm{RF}}\overline{\beta }\overline{\mathbf{w}}_{\mathrm{eq,}k}x_{k}}}\notag \\
&+\underset{\mathrm{interference}}{\underbrace{\widetilde{\mathbf{\omega }}_{k}^{H}\mathbf{H}_{k}^{T}\overset{N}{\underset{j=1,j\neq k}{\sum }}\mathbf{F}%
_{\mathrm{RF}}\overline{\beta }\overline{\mathbf{w}}_{\mathrm{eq,}j}x_{j}}}+%
\underset{\mathrm{noise}}{\underbrace{\widetilde{\mathbf{\omega }}_{k}^{H}%
\mathbf{z}_{\mathrm{MS},k}}},  \label{P2}
\end{align}\vspace*{-0.0mm}%
where $x_{k}\in\mathbb{C}^{1\times 1}$ is the transmitted symbol from the BS to user $k$ in the desired cell, $\mathrm{E}\left[ \left\vert x_{k}^{2}\right\vert \right]=E_{s}$, $E_{s}$ is the average transmitted symbol energy for each user, $\overline{\beta }=\sqrt{\tfrac{1}{\mathrm{tr}(\overline{\mathbf{W}}%
_{\mathrm{eq}}\overline{\mathbf{W}}_{\mathrm{eq}}^{H})}}$ is the
transmission power normalization factor, and the effective
noise part $\mathbf{z}_{\mathrm{MS,}%
k}\sim \mathcal{CN}\left( \mathbf{0},{\sigma_{\mathrm{MS}}^{2}}\mathbf{I}\right) $.
Then we express the signal-to-interference-plus-noise ratio (SINR) of user $k$ as \vspace{-0mm}
\begin{equation}
\mathrm{SINR}_{\mathrm{ZF}}^{k}=\frac{\overline{\beta }^{2}E_{s}}{\sigma _{%
\mathrm{MS}}^{2}}.  \label{Eq_1520}
\end{equation}\vspace{-0mm}%
In the sequel, we study the performance of the considered hybrid mmWave system. For simplicity, we assume the mmWave channels of all the users have the same Rician K-factor, i.e., $\upsilon _{k} = \upsilon, \forall k$.
\vspace*{+0.0mm}
\subsection{Performance Upper Bound of ZF Precoding}

Now, exploiting the SINR expression in (\ref{Eq_1520}),  we summarize the upper bound of achievable rate per user of the ZF precoding in a theorem at the top of this page.\vspace*{-0mm}
From Equation (\ref{Theo_1}), we see that the upper bound of achievable rate per user of the proposed hybrid ZF precoding depends on the Rician K-factor, $ \upsilon$.
We can further observe that the upper bound of the achievable rate per user also depends on the BS beamforming matrix $\mathbf{F}_{\mathrm{RF}}$ designed in step $2$ of the proposed CSI estimation.
With an increasing number of antennas at the BS, communication channels are more likely to be orthogonal.
Therefore, it is interesting to evaluate the asymptotic upper bound $R_{\mathrm{HB}}^{\mathrm{upper}}$ for the case of a large number of antennas. We note that, even if the number of antennas equipped at the BS is sufficiently large, the required number of RF chains is still equal to the number of users in the hybrid mmWave systems and the result is summarized in Corollary \ref{Coro_1} at the top of this page.
In Equation (\ref{HSUB_LA}), we have the intuitive observation that the performance of the  proposed hybrid precoding is mainly determined by the equipped numbers of antennas and RF chains.

\subsection{Performance of Fully Digital System}
\begin{figure}[t]
\centering
\includegraphics[width =
3.5in]{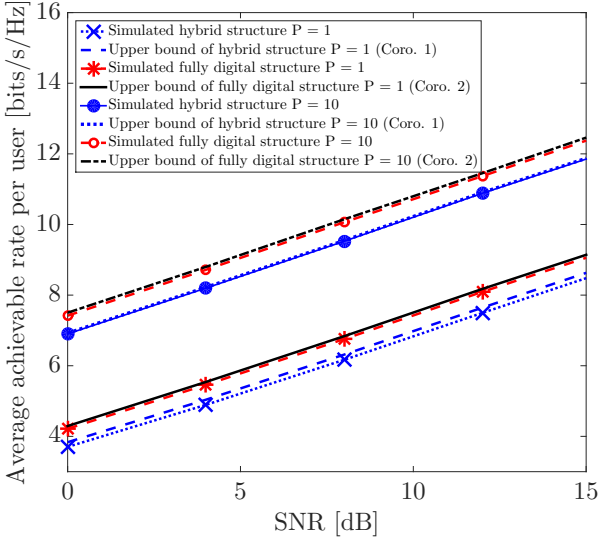}
\caption{The average achievable rate per user [bits/s/Hz] versus SNR for the hybrid system and the fully digital system.}
\label{fig:HBvsFD}
\end{figure}
In this section, we derive the achievable rate performance of a fully digital mmWave system in the large numbers of antennas regime. The obtained analytical results in this section will be used to  compare with the considered hybrid system in the simulation section.
To this end, for the fully digital mmWave system, we assume that each user is equipped with one RF chain and $P$ antennas.
The $P$ antenna array equipped at each user can provide $10\log_{10}(P)$ dB array gain.
We note that, the number of antennas equipped at the BS is $M$ and the number of RF chains equipped at the BS is equal to the number of antennas.
The channel matrix for user $k$ is given by\vspace*{-0mm}
\setcounter{equation}{21}
\begin{equation}
\mathbf{H}_{k}^{T}=\mathbf{h}_{k}^{\ast }\mathbf{h}_{\mathrm{BS,}k}^{T}.
\end{equation}\vspace*{-0mm}%
We assume that the CSI is perfectly known to the users and the BS is with the fully digital system to illustrate the maximal performance gap between the proposed structure and the perfect case.
Therefore, the achievable rate per user upper bound of the fully digital system is summarized in the following Corollary 2.
\vspace*{-0mm}
\begin{coro}
In the large numbers of antennas regime, the asymptotic achievable rate per user of the fully digital system is bounded above by
\begin{equation}
R_{\mathrm{FD}}\leqslant R_{\mathrm{FD}}^{\mathrm{upper}}\underset{%
M\rightarrow \infty }{\overset{a.s.}{\rightarrow }}\log _{2}\left[ 1+\frac{MP%
}{N}\dfrac{E_{s}}{\sigma _{\mathrm{MS}}^{2}}\right] .  \label{FDUB_LA}
\end{equation}
\end{coro}
\begin{proof}
The result follows similar procedures the proof as in Appendix A.
\end{proof}

In the large numbers of antennas regime, based on (\ref{HSUB_LA}) and (\ref{FDUB_LA}), it is interesting to observe that with an increasing Rician K-factor $\upsilon $, the performance upper bounds of the two considered structures will coincide.\vspace{-0mm}%

\vspace*{+0.0mm}
\section{Simulation and Discussion}

In this section, we present numerical results to validate our analysis.
We consider a single cell hybrid mmWave system.
\begin{figure}[t]
\centering
\includegraphics[width =3.5in]{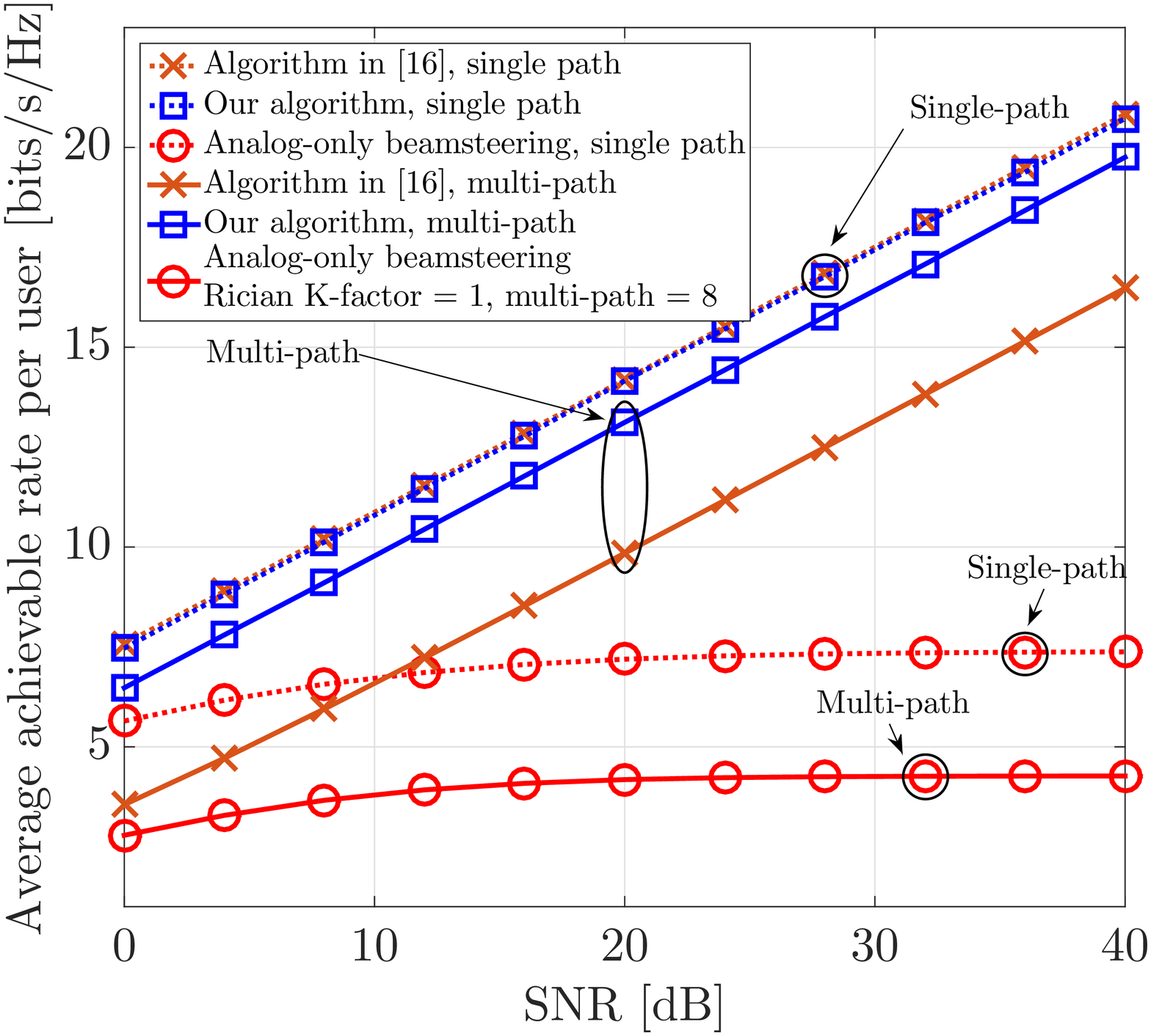}
\caption{The average achievable rate per user (bits/s/Hz) versus SNR for \cite{Alkhateeb2015}
proposed algorithm and our proposed algorithm.}
\label{fig:Comp_HB_ALH_ours}
\end{figure}
In Figure \ref{fig:HBvsFD}, we present a comparison between the achievable rate per user of the hybrid system and the fully digital system for $M=100,$ $N=10$, and a Rician K-factor of $\upsilon_{k} =2,\forall k$.
First, our simulation results verify the tightness of derived upper bounds in (\ref{HSUB_LA}) and (\ref{FDUB_LA}).
It can be observed from Figure \ref{fig:HBvsFD} that, even for a small value of Rician K-factor, our proposed channel estimation scheme with ZF precoding can achieve considerable high sum rate performance due to its interference suppression capability.
In addition, the performance gap between the fully digital system and the hybrid system is small, which is determined by the strongest AoA component to the scattering component ratio.

In Figure \ref{fig:Comp_HB_ALH_ours}, we illustrate the effectiveness of the proposed non-sparse mmWave channel estimation algorithm.
We assume perfect channel estimation with $M=100,$ $N=4,$ and $P=16$. For non-sparse mmWave channels, we assume $\upsilon_{k} =1,\forall k$.
In Figure \ref{fig:Comp_HB_ALH_ours}, we compare between the achievable rates using the proposed hybrid algorithm and the algorithm proposed by \cite{Alkhateeb2015} for sparse and non-sparse mmWave channels.
For sparse single-path channels, the achievable rate of the proposed algorithm matches with the algorithm proposed in \cite{Alkhateeb2015}.
For non-sparse mmWave channels, with the number of multi-paths $N_{\mathrm{l}}=8$, we observe that the proposed algorithm achieves a better system performance than that of the algorithm proposed in \cite{Alkhateeb2015}.
The reason is that, the proposed algorithm takes the scattering components into account and exploits the strongest AoAs of all the users to suppress the MU interference.
In contrast, the algorithm proposed in \cite{Alkhateeb2015}, which aims to maximize the desired signal energy, does not suppress the MU interference as effective as our proposed algorithm.
Furthermore, Figure \ref{fig:Comp_HB_ALH_ours} also illustrates that a significant achievable rate gain is brought by the proposed channel estimation and ZF precoding over a simple analog-only beamforming steering scheme.

\vspace*{+0.0mm}
\section{Conclusions}

In this paper, we proposed a low-complexity mmWave channel estimation for the MU hybrid mmWave systems, which is applicable for both sparse and non-sparse mmWave channel environments.
The achievable rate performance of ZF precoding based on the proposed channel estimation scheme was derived and compared with the achievable rate of fully digital systems.
The analytical and simulation results indicated that the proposed scheme can approach the rate performance achieved by the fully digital system with sufficient large Rician K-factors.

\vspace*{+0.0mm}
\section*{Appendix}
\vspace{-0mm}
\subsection{Proof of Theorem 1}

The average achievable rate per user of ZF precoding is given by \vspace*{-0mm}%
\begin{equation}
R_{\mathrm{HB}}=\mathrm{E}_{%
\mathrm{H}_{\mathrm{S}}}\left\{ \log _{2}\left[ 1+\left[ \mathrm{tr}\left[ (%
\mathbf{H}_{\mathrm{eq}}^{T}\mathbf{H}_{\mathrm{eq}}^{\ast })^{-1}\right] %
\right] ^{-1}\dfrac{E_{s}}{\sigma _{\mathrm{MS}}^{2}}\right] \right\} .
\label{Proof_1}
\end{equation}\vspace*{-0mm}%
First, we introduce some preliminaries. Since $\mathbf{H}_{\mathrm{eq}}^{T}\mathbf{H}_{\mathrm{eq}}^{\ast}$ is a positive definite Hermitian matrix, by eigenvalue decomposition, it can be decomposed as $\mathbf{H}_{\mathrm{eq}}^{T}\mathbf{H}_{\mathrm{eq}}^{\ast }=\mathbf{U\Lambda V}^{H}$, $\mathbf{\Lambda}\in\mathbb{C}^{N\times N}$ is the positive diagonal eigenvalue matrix, while $\mathbf{V}\in\mathbb{C}^{N\times N}$ and $\mathbf{U}\in\mathbb{C}^{N\times N}$ are unitary matrixes, $\mathbf{U=V}^{H}$.
The sum of the eigenvalues of $\mathbf{H}_{\mathrm{eq}}^{T}\mathbf{H}_{\mathrm{eq}}^{\ast}$ equals to the trace of matrix $\mathbf{\Lambda }$.
Then we can rewrite the power normalization factor in (\ref{Proof_1}) as\vspace*{-2mm}%
\begin{equation}
\frac{N}{\mathrm{tr}\left[ (\mathbf{H}_{\mathrm{eq}}^{T}\mathbf{H}_{\mathrm{eq}}^{\ast })^{-1}\right] }=\left[ \overset{N}{\underset{i=1}{\dsum }}\frac{1}{N}\lambda _{i}^{-1}\right]^{-1},\label{Proof_2}
\end{equation}\vspace{-0mm}%
In addition, $f(x)=x^{-1}$ is a strictly decreasing convex function and exploiting the convexity, we have the following results \cite{book:infotheory}:\vspace{-0mm}%
\begin{equation}
\left[ \overset{N}{\underset{i=1}{\dsum }}\frac{1}{N}\lambda _{i}^{-1}\right]
^{-1}\leqslant \overset{N}{\underset{i=1}{\dsum }}\frac{1}{N}\left[ \left(
\lambda _{i}^{-1}\right) ^{-1}\right] =\overset{N}{\underset{i=1}{\dsum }}%
\frac{1}{N}\lambda _{i}.  \label{JIE}
\end{equation}\vspace{-0mm}%
Therefore, based on (\ref{Proof_2}) and (\ref{JIE}), we have the following inequality: \vspace{-0mm}%
\begin{equation}
\frac{1}{\mathrm{tr}\left[ \left( \mathbf{H}_{\mathrm{eq}}^{T}\mathbf{H}_{%
\mathrm{eq}}^{\ast }\right) ^{-1}\right] }\leqslant \overset{N}{\underset{i=1%
}{\dsum }}\frac{1}{N^{2}}\lambda _{i}=\frac{1}{N^{2}}\mathrm{tr}\left[ \mathbf{H}_{%
\mathrm{eq}}^{T}\mathbf{H}_{\mathrm{eq}}^{\ast }\right] .  \label{Proof_3}
\end{equation}\vspace{-0mm}%
From (\ref{Proof_3}), Equation (\ref{Proof_1}) can be rewritten as (\ref{Theo_1}) in Theorem 1.\vspace{-0mm}%


\begin{thebibliography}{10}
\providecommand{\url}[1]{#1}
\csname url@samestyle\endcsname
\providecommand{\newblock}{\relax}
\providecommand{\bibinfo}[2]{#2}
\providecommand{\BIBentrySTDinterwordspacing}{\spaceskip=0pt\relax}
\providecommand{\BIBentryALTinterwordstretchfactor}{4}
\providecommand{\BIBentryALTinterwordspacing}{\spaceskip=\fontdimen2\font plus
\BIBentryALTinterwordstretchfactor\fontdimen3\font minus
  \fontdimen4\font\relax}
\providecommand{\BIBforeignlanguage}[2]{{%
\expandafter\ifx\csname l@#1\endcsname\relax
\typeout{** WARNING: IEEEtran.bst: No hyphenation pattern has been}%
\typeout{** loaded for the language `#1'. Using the pattern for}%
\typeout{** the default language instead.}%
\else
\language=\csname l@#1\endcsname
\fi
#2}}
\providecommand{\BIBdecl}{\relax}
\BIBdecl

\bibitem{Kwan_5G}
V.~W.~S. Wong, R.~Schober, D.~W.~K. Ng, and L.-C. Wang, \emph{{Key Technologies
  for 5G Wireless Systems}}.

\bibitem{AZhang2015}
J.~A. Zhang, X.~Huang, V.~Dyadyuk, and Y.~J. Guo, ``Massive hybrid antenna
  array for millimeter-wave cellular communications,'' \emph{IEEE Wireless
  Commun.}, vol.~22, no.~1, pp. 79--87, Feb. 2015.

\bibitem{Dai2016}
\BIBentryALTinterwordspacing
L.~Dai, X.~Gao, S.~Han, C.~L. I, and X.~Wang, ``Beamspace channel estimation
  for millimeter-wave massive {MIMO} systems with lens antenna array,'' 2016.
  [Online]. Available: \url{http://arxiv.org/abs/1607.05130v1}
\BIBentrySTDinterwordspacing

\bibitem{Kokshoorn2016}
\BIBentryALTinterwordspacing
M.~Kokshoorn, H.~Chen, P.~Wang, Y.~Li, and B.~Vucetic, ``{Millimeter wave MIMO
  channel estimation using overlapped beam patterns and rate adaptation},''
  2016. [Online]. Available: \url{https://arxiv.org/abs/1603.01926v2}
\BIBentrySTDinterwordspacing

\bibitem{JR:Kwan_massive_MIMO}
D.~W.~K. Ng, E.~S. Lo, and R.~Schober, ``{Energy-Efficient Resource Allocation
  in OFDMA Systems with Large Numbers of Base Station Antennas},'' \emph{IEEE
  Trans. Commun.}, vol.~11, no.~9, pp. 3292--3304, 2012.

\bibitem{Yang2015}
N.~Yang, L.~Wang, G.~Geraci, M.~Elkashlan, J.~Yuan, and M.~D. Renzo,
  ``{Safeguarding 5G wireless communication networks using physical layer
  security},'' \emph{IEEE Commun. Mag.}, vol.~53, no.~4, pp. 20--27, Apr. 2015.

\bibitem{Bogale2015}
T.~E. Bogale and L.~B. Le, ``{Massive MIMO and millimeter wave for 5G wireless
  HetNet: Potentials and challenges},'' \emph{IEEE Veh. Technol. Mag.},
  vol.~11, no.~1, pp. 64--75, Mar. 2016.

\bibitem{Marzetta2010}
T.~L. Marzetta, ``{Noncooperative cellular wireless with unlimited numbers of
  base station antennas},'' \emph{IEEE Trans. Wireless Commun.}, vol.~9,
  no.~11, pp. 3590--3600, Nov. 2010.

\bibitem{Swindlehurst2014}
A.~L. Swindlehurst, E.~Ayanoglu, P.~Heydari, and F.~Capolino,
  ``{Millimeter-wave massive MIMO: the next wireless revolution?}'' \emph{IEEE
  Commun. Mag.}, vol.~52, no.~9, pp. 56--62, Sept. 2014.

\bibitem{Deng2015}
Y.~Deng, L.~Wang, K.~K. Wong, A.~Nallanathan, M.~Elkashlan, and S.~Lambotharan,
  ``Safeguarding massive {MIMO} aided hetnets using physical layer security,''
  in \emph{Intern. Conf. on Wireless Commun. Signal Process. (WCSP)}, Oct.
  2015, pp. 1--5.

\bibitem{Sohrabi2016}
F.~Sohrabi and W.~Yu, ``Hybrid digital and analog beamforming design for
  large-scale antenna arrays,'' \emph{IEEE J. Select. Topics in Signal
  Process.}, vol.~10, no.~3, pp. 501--513, Apr. 2016.

\bibitem{Rappaport2015}
T.~S. Rappaport, G.~R. MacCartney, M.~K. Samimi, and S.~Sun, ``Wideband
  millimeter-wave propagation measurements and channel models for future
  wireless communication system design,'' \emph{IEEE Trans. Commun.}, vol.~63,
  no.~9, pp. 3029--3056, Sept. 2015.

\bibitem{Bjornson2016}
E.~Bj{\"{o}}rnson, E.~G. Larsson, and T.~L. Marzetta, ``{Massive MIMO: Ten
  myths and one critical question},'' \emph{IEEE Commun. Mag.}, vol.~54, no.~2,
  pp. 114--123, Feb. 2016.

\bibitem{Heath2016a}
R.~W. Heath, N.~G. Prelcic, S.~Rangan, W.~Roh, and A.~M. Sayeed, ``{An overview
  of signal processing techniques for millimeter wave MIMO systems},''
  \emph{IEEE J. Select. Topics in Signal Process.}, vol.~10, no.~3, pp.
  436--453, April 2016.

\bibitem{Ni2016}
W.~Ni and X.~Dong, ``Hybrid block diagonalization for massive multiuser {MIMO}
  systems,'' \emph{IEEE Trans. Commun.}, vol.~64, no.~1, pp. 201--211, Jan.
  2016.

\bibitem{Alkhateeb2015}
A.~Alkhateeb, G.~Leus, and R.~W. Heath, ``Limited feedback hybrid precoding for
  multi-user millimeter wave systems,'' \emph{IEEE Trans. Wireless Commun.},
  vol.~14, no.~11, pp. 6481--6494, Nov. 2015.

\bibitem{Ayach2014}
O.~E. Ayach, S.~Rajagopal, S.~Abu-Surra, Z.~Pi, and R.~W. Heath, ``Spatially
  sparse precoding in millimeter wave {MIMO} systems,'' \emph{IEEE Trans.
  Wireless Commun.}, vol.~13, no.~3, pp. 1499--1513, Mar. 2014.

\bibitem{Han2015}
S.~Han, C.~l.~I, Z.~Xu, and C.~Rowell, ``Large-scale antenna systems with
  hybrid analog and digital beamforming for millimeter wave 5{G},'' \emph{IEEE
  Commun. Mag.}, vol.~53, no.~1, pp. 186--194, Jan. 2015.

\bibitem{Hur2016}
S.~Hur, S.~Baek, B.~Kim, Y.~Chang, A.~F. Molisch, T.~S. Rappaport, K.~Haneda,
  and J.~Park, ``Proposal on millimeter-wave channel modeling for 5{G} cellular
  system,'' \emph{IEEE J. Select. Topics in Signal Process.}, vol.~10, no.~3,
  pp. 454--469, Apr. 2016.

\bibitem{Alkhat2014}
A.~Alkhateeb, O.~E. Ayach, G.~Leuz, and R.~W. Heath, ``{Channel estimation and
  hybrid precoding for millimeter wave cellular systems},'' \emph{IEEE J. Sel.
  Topics in Signal Process.}, vol.~8, no.~5, pp. 831--846, Oct. 2014.

\bibitem{Buzzi2016}
\BIBentryALTinterwordspacing
S.~Buzzi and C.~D'Andrea, ``{Doubly massive mmWave MIMO systems: Using very
  large antenna arrays at both transmitter and receiver},'' 2016. [Online].
  Available: \url{https://arxiv.org/abs/1607.07234v1}
\BIBentrySTDinterwordspacing

\bibitem{Al-Daher2012}
Z.~Al-Daher, L.~P. Ivrissimtzis, and A.~Hammoudeh, ``Electromagnetic modeling
  of high-frequency links with high-resolution terrain data,'' \emph{IEEE
  Antennas and Wireless Propagation Lett.}, vol.~11, pp. 1269--1272, Oct. 2012.

\bibitem{book:wireless_comm}
D.~Tse and P.~Viswanath, \emph{{Fundamentals of wireless communication}}.\hskip
  1em plus 0.5em minus 0.4em\relax {Cambridge University Press}, 2005.

\bibitem{Trees2002}
H.~L.~V. Trees, \emph{{Optimum array processing: {Part} IV of detection,
  estimation, and modulation theory}}.\hskip 1em plus 0.5em minus 0.4em\relax
  John Wiley \& Sons, Inc., 2002.

\bibitem{Yang2013}
H.~Yang and T.~L. Marzetta, ``{Performance of conjugate and zero-forcing
  beamforming in large-scale antenna systems},'' \emph{IEEE J. Select. Areas
  Commun.}, vol.~31, no.~2, pp. 172--179, Mar. 2013.

\bibitem{book:infotheory}
T.~M. Cover and J.~A. Thomas, \emph{{Elements of Information Theory}}.\hskip
  1em plus 0.5em minus 0.4em\relax {New York: Wiley}, 1991.

\end{thebibliography}

\bibliographystyle{IEEEtran}

\end{document}